%% file: main.tex
\documentclass[a4paper,twocolumn,10pt]{quantumarticle}
\pdfoutput=1

\usepackage[utf8]{inputenc}
\usepackage[english]{babel}
\usepackage[T1]{fontenc}
\usepackage{amsmath}
\usepackage{hyperref}
\usepackage{tikz}
\usepackage{wrapfig}
\usepackage{cite}
\usepackage{amsmath,amssymb,amsfonts}
\usepackage{graphicx}
\usepackage{textcomp}
\usepackage{xcolor}
\usepackage{mathtools}
\usepackage{amsthm}
\newtheorem{theorem}{Theorem}
\usepackage[ruled,noend]{algorithm2e}
\usepackage{microtype}
\usepackage{float}
\usepackage{adjustbox}
\usepackage{booktabs,makecell,tabularx}
\usepackage{url}
\usepackage{booktabs}
\usepackage{caption}
\usepackage{subcaption}
\usepackage[braket, qm]{qcircuit}
\usepackage{braket}
\usepackage{tcolorbox}
\usepackage[frozencache]{minted}

\setminted[python]{
fontsize = \footnotesize,
}

\tcbset{colback=black!5!white,arc=2pt,boxrule=0.5pt}
\BeforeBeginEnvironment{minted}{\begin{tcolorbox}}
\AfterEndEnvironment{minted}{\end{tcolorbox}}

\theoremstyle{definition}
\newtheorem{definition}{Definition}[section]

\begin{document}

\title{Automatic quantum function parallelization and memory management in Qrisp}

\author{Raphael Seidel}
\email{raphael.seidel@iqm.tech}
\thanks{This work was originally published at \href{https://quantum-compilers.github.io/iwqc2024/}{IWQC 2024} while working at the Fraunhofer Institute for Open Communication Systems. At the time of ArXiv upload, the author works at IQM Quantum Computers.}
\affiliation{Fraunhofer Institute for Open Communication Systems, Berlin, Germany}
\affiliation{IQM Quantum Computers, Munich, Germany}

\maketitle

\begin{abstract}
Automated optimization of quantum programs has gathered significant attention amidst the recent advances of hardware manufacturers. In this work we introduce a novel data-structure for representing quantum programs called permeability DAG, which captures several useful properties of quantum programs across multiple levels of abstraction. Operating on this representation facilitates a variety of powerful transformations such as automatic parallelization, memory management and synthesis of uncomputation. More potential use-cases are listed in the outlook section. At the core, our representation abstracts away a class of non-trivial commutation relations, which stem from a feature called permeability. Both memory management and parallelization can be made sensitive to execution speed details of each particular quantum gate, implying our compilation methods are not only retargetable between NISQ/FT but even for individual device instances.
\end{abstract}

\input{sections/introduction}
\input{sections/permeability}
\input{sections/dag}
\input{sections/parallelization}
\input{sections/memory_management}
\input{sections/outlook}
\input{sections/summary}

\section*{Code availability}
\textit{Qrisp} is an open-source Python framework for high-level programming of quantum computers.
The source code is available in \href{https://github.com/eclipse-qrisp/Qrisp}{https://github.com/eclipse-qrisp/Qrisp}.

\bibliographystyle{quantum}
\bibliography{sources}

\appendix
\input{sections/appendix}

\end{document}

%% file: sections/introduction.tex
\section{Introduction}
\label{sec:introduction}
The promising advances of recent quantum hardware manufacturers \cite{QuEra_2023, quantinuum_24, oxford_ionics_2024} spark realistic hopes of achieving large scale fault tolerant quantum computations within less than a decade. As the scale of treatable problems grows, the implementation of more elaborate algorithms faces a variety of challenges. One of them is the lack of systematic programming abstractions, which are indispensable to all of modern software engineering (such as modular development, systematic testing or code introspection for debugging). Within a recent work \cite{seidel_2024_qrisp} we gave an overview over our quantum programming framework \textit{Qrisp}, which tackles these challenges. A core advantage of \textit{Qrisp} is the possibility to modularize algorithm development, which is rooted in the feature of automatic quantum memory management. The present work introduces the structures and algorithms behind this feature. Additionally, we present a related algorithm, which automatically reduces the circuit runtime by establishing parallelism between quantum function calls.
\begin{figure}
    \centering
    \includegraphics[width=\linewidth]{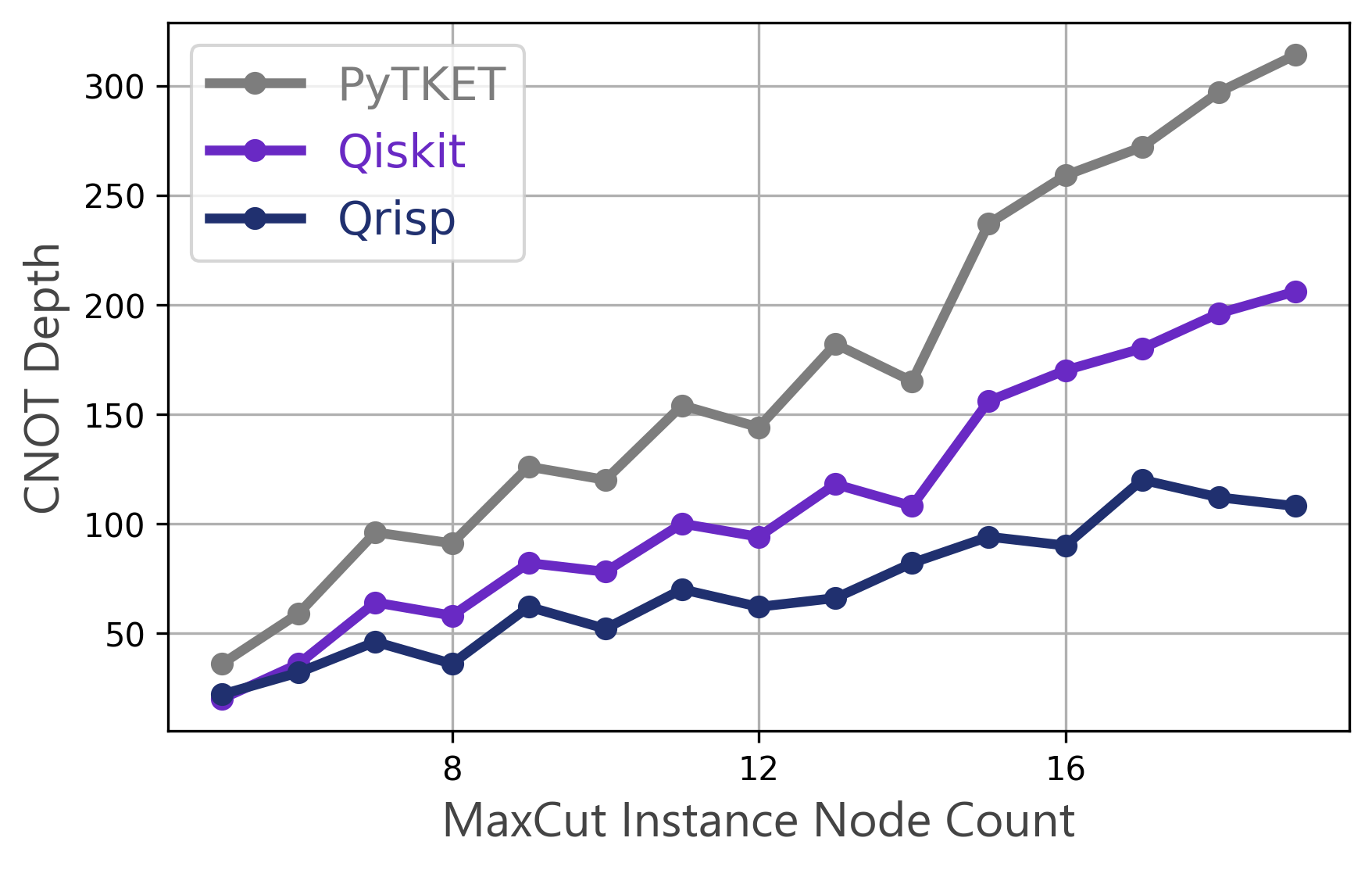}
    \caption{A plot of the resulting CNOT depth after applying several circuit optimizers (including the presented parallelization algorithm) to MaxCut problem circuits \cite{farhi_2014}. The problem instances are constructed by sampling a random Erdös–Rényi graph \cite{gilbert_1959} with $p = 0.5$ for each node count. To optimize the circuits in Qiskit, we leverage the \texttt{transpile} function with optimization level 3. For PyTKET, we used the \texttt{FullPeepholeOptimise} optimization pass, which is advertised as a one fits all procedure. We note that our method works particularly well for MaxCut problems, implying general algorithms will most likely see a more moderate improvement in circuit depth.}
    \label{fig:parallelization_plot}
\end{figure}
Both of these algorithms are based on a novel DAG representation called \textit{permeability DAG}, which leverages non-trivial commutation relations among quantum function calls. In its essence, the permeability DAG captures information about which functions leave their inputs ``constant'' in a rigorously defined notion of ``constant'', which we call \textit{permeability}. We show that if there are several functions operating in a permeable way on the same set of qubits, these functions can be interchanged. The permeability DAG therefore represents quantum algorithms with the mentioned commutation relations hard coded into its internal structure. As such, a much broader class of transformations is possible by operating on this representation.\\
A subset of the mentioned type of commutation relations have already been realized in the form of a DAG within the context of automatic synthesis of uncomputation circuits \cite{paradis_2021} (``Unqomp''). In this work we refine the Unqomp DAG to facilitate more general uncomputations and also the above mentioned compilation algorithms for memory management and parallelization.

%% file: sections/permeability.tex
\section{Permeability}
\label{sec:permeability}
Permeability is a property of (composite) quantum gates, that was initially introduced by us in \cite{Seidel_2023_uncomputation} and gives a formal meaning of what ``constant'' could mean for a quantum function. Even though the original purpose of the concept was restricted to automatic synthesis of uncomputations, we realized that the scope of applications is much broader. At the core, permeability information allows the compiler to leverage non-trivial commutation relations. To understand how this works in detail, we recall the definition given in \cite{Seidel_2023_uncomputation}.

\begin{tcolorbox}
\begin{definition}[Permeability]
A (composite) quantum gate $U$ is called Z-permeable in qubit $i$, if it commutes with the $Z$ operator on this qubit.
\begin{align}
   \text{U} \text{ is permeable on qubit i} \Leftrightarrow \text{U} \text{Z}_i = \text{Z}_i \text{U}
\end{align}
Similarly, $U$ is called X-permeable in qubit $i$ if it commutes with the $X$ operator on this qubit.
\end{definition}
\end{tcolorbox}

For Z-permeability we gave a proof of the following theorem in \cite{Seidel_2023_uncomputation}. For the readers convenience, it can be found in the appendix. In this work we expand the statement to X-permeability.
\begin{tcolorbox}
\begin{theorem}[Commutativity theorem]
\label{commutativity_theorem}
Let $U$ and $V$ be $n$ and $m$ qubit operators, respectively. If $U$ is Z-permeable (X-permeable) in its last $p$ qubits and $V$ is Z-permeable (X-permeable) in its first $p$ qubits, then the two operators commute if they intersect only on these qubits:
\begin{align}
\begin{aligned}
(U \otimes \mathbbm{1}^{\otimes m-p})& (\mathbbm{1}^{\otimes n-p} \otimes V)\\
&=\\
(\mathbbm{1}^{\otimes n-p} \otimes V)& (U \otimes \mathbbm{1}^{\otimes m-p})
\end{aligned}
\end{align}
\end{theorem}
\end{tcolorbox}
The proof for the statement in the case of X-permeability can be found in Appendix \ref{sec:x_com_proof}.\\
According to the results proven in \cite{Seidel_2023_uncomputation}, determining the permeability status for arbitrary gates is possible in linear time (in the unitary size) if the unitary is known. If the gate $U$ is composite and only acting via Z-permeable (X-permeable) gates on the qubit $q$, Z-permeability (X-permeability) can be inferred without the unitary. Furthermore, in \textit{Qrisp} it is possible for the programmer to annotate function definitions regarding their permeability status, implying in practice permeability information very rarely has to be inferred via unitary calculation.
To highlight the prevalence of permeability features, we now list some common gates, which are permeable.
\begin{itemize}
    \item Any RZ (RX) gate is Z-permeable (X-permeable).
    \item Z, P, CZ, CP, CRZ, RZZ gates are Z-permeable in all qubits.
    \item Any Pauli-X-gadget \cite{Cowtan_2020} ($U(\phi) = \text{exp}(i\phi \otimes_{i=0}^n X_i)$) is X-permeable in all qubits. The same applies to Pauli-Z-gadgets.
    \item Global M\o{}lmer-S\o{}rensen gates \cite{Maslov_2018} are X-permeable (or Z depending on the hardware \cite{W_lk_2017}) in all qubits.
    \item Any phase polynomial \cite{Amy_2018} and any diagonal Hamiltonian evolution is Z-permeable in all inputs.
    \item Any (multi-)controlled gate is Z-permeable in any control qubit. If the base gate is Z-permeable (X-permeable), the controlled version is also Z-permeable (X-permeable) on that qubit.
    \item Any gate $U_f$ that realizes an out-of-place classical function $f$ in superposition, by acting as 
    \begin{align}
    U_f \ket{x_0}\dotsb\ket{x_n}\ket{0} = \ket{x_1}\dotsb\ket{x_n}\ket{f(x_0,\dotsc,x_n)}
    \end{align}
    is Z-permeable in all input qubits. This especially applies to all arithmetic functions.
    \item Quantum-quantum (modular) in-place adders are permeable in the input that is not operated on.
\end{itemize}

%% file: sections/dag.tex
\section{The permeability DAG}
\label{sec:dag}
DAGs (directed acyclic graphs) have found manifold application in quantum compilation \cite{Joseph_2023, promponas_2024, meuli_2022}. In their essence, DAGs are so useful because they can capture equivalence classes of reorderings of a given sequence.
The DAG representation we construct here is based on the ideas presented in the Unqomp \cite{paradis_2021}, however enriched by the concepts of Z/X-permeability. This enables the permeability DAG to express equivalence classes of quantum circuits, that are equivalent according to the commutation relations induced by Theorem~\ref{commutativity_theorem}. To leverage these, we extract circuit reorderings (with an invariant unitary) by determining a topological sort of the DAG. Topological sorting algorithms are however a wide class of techniques, which allow us to select a procedure, that favors certain characteristics of the circuit (like low circuit depth).
\\
\begin{figure*}[t!]
    \centering
    \begin{subfigure}[b]{0.5\textwidth}
        \centering
        \scalebox{1}{
        \Qcircuit @C=1.0em @R=0.2em @!R { \\
        	 	\nghost{{q0} :  } & \lstick{{q_0} :  } & \ctrl{1} & \ctrl{2} & \ctrl{3} & \qw & \qw\\
        	 	\nghost{{q1} :  } & \lstick{{q_1} :  } & \gate{\mathrm{Y}} & \qw & \qw & \qw & \qw\\
        	 	\nghost{{q2} :  } & \lstick{{q_2} :  } & \qw & \gate{\mathrm{Y}} & \qw & \qw & \qw\\
        	 	\nghost{{q3} :  } & \lstick{{q_3} :  } & \qw & \qw & \gate{\mathrm{Y}} & \qw & \qw\\
        \\ }}
        \caption{\label{fig:simple_circuit}}
    \end{subfigure}%
    \begin{subfigure}[b]{0.5\textwidth}
        \centering
        \scalebox{1.0}{
        \Qcircuit @C=1.0em @R=0.2em @!R { \\
        	 	\nghost{{q0} :  } & \lstick{{q0} :  } & \ctrl{1} & \ctrl{2} & \ctrl{3} & \targ & \gate{\mathrm{X}} & \targ & \qw & \qw\\
        	 	\nghost{{q1} :  } & \lstick{{q1} :  } & \gate{\mathrm{Y}} & \qw & \qw & \ctrl{-1} & \qw & \qw & \qw & \qw\\
        	 	\nghost{{q2} :  } & \lstick{{q2} :  } & \qw & \gate{\mathrm{Y}} & \qw & \qw & \qw & \ctrl{-2} & \qw & \qw\\
        	 	\nghost{{q3} :  } & \lstick{{q3} :  } & \qw & \qw & \gate{\mathrm{Y}} & \qw & \qw & \ctrl{-1} & \qw & \qw\\
        \\ }}
        \caption{\label{fig:less_simple_circuit}}
    \end{subfigure}
        \begin{subfigure}[b]{0.5\textwidth}
        \centering
        \includegraphics[width = 0.65\textwidth]{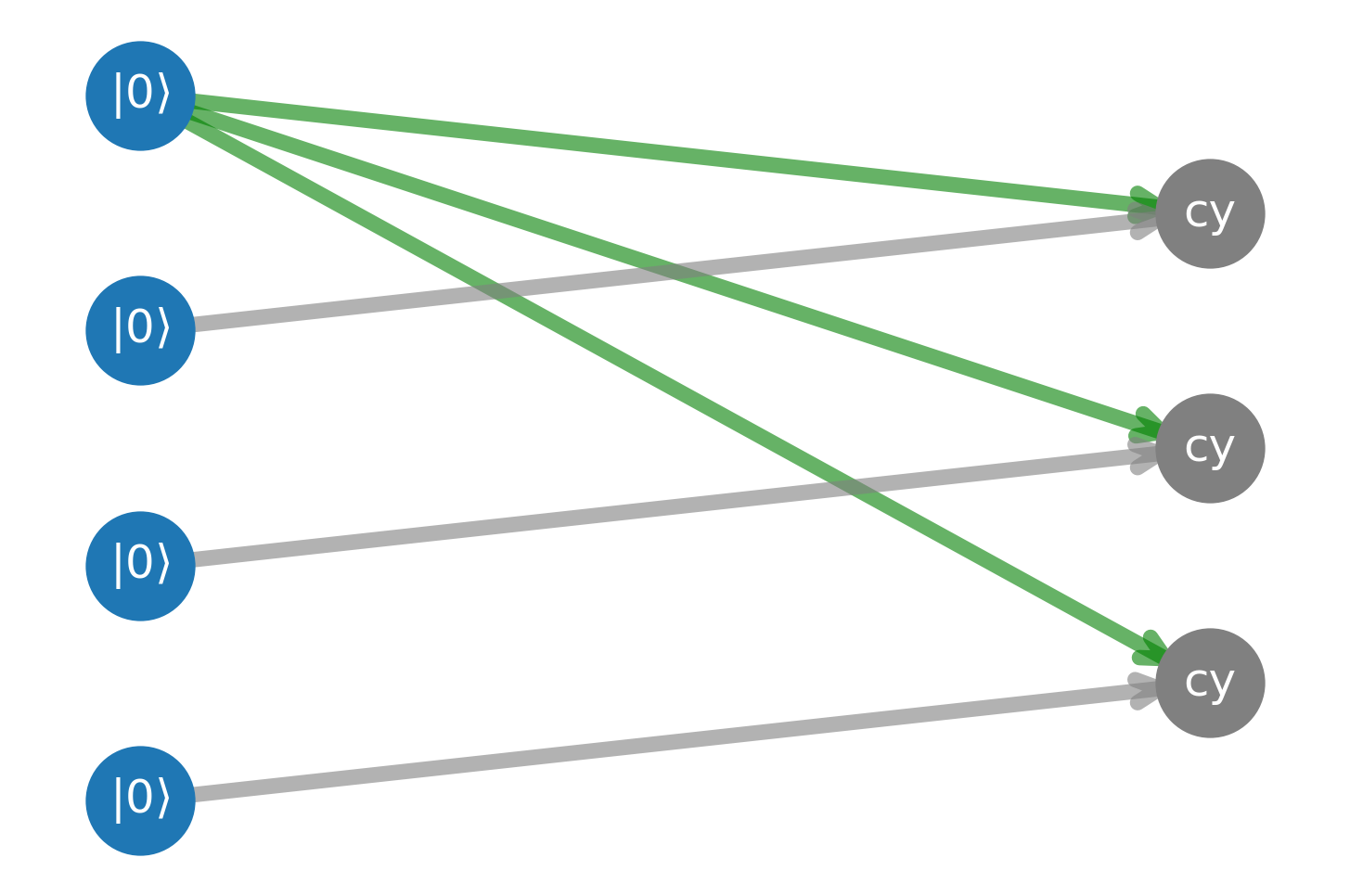}
        \caption{\label{fig:simple_dag}}
    \end{subfigure}%
        \begin{subfigure}[b]{0.5\textwidth}
        \centering
        \includegraphics[width = 0.65\textwidth]{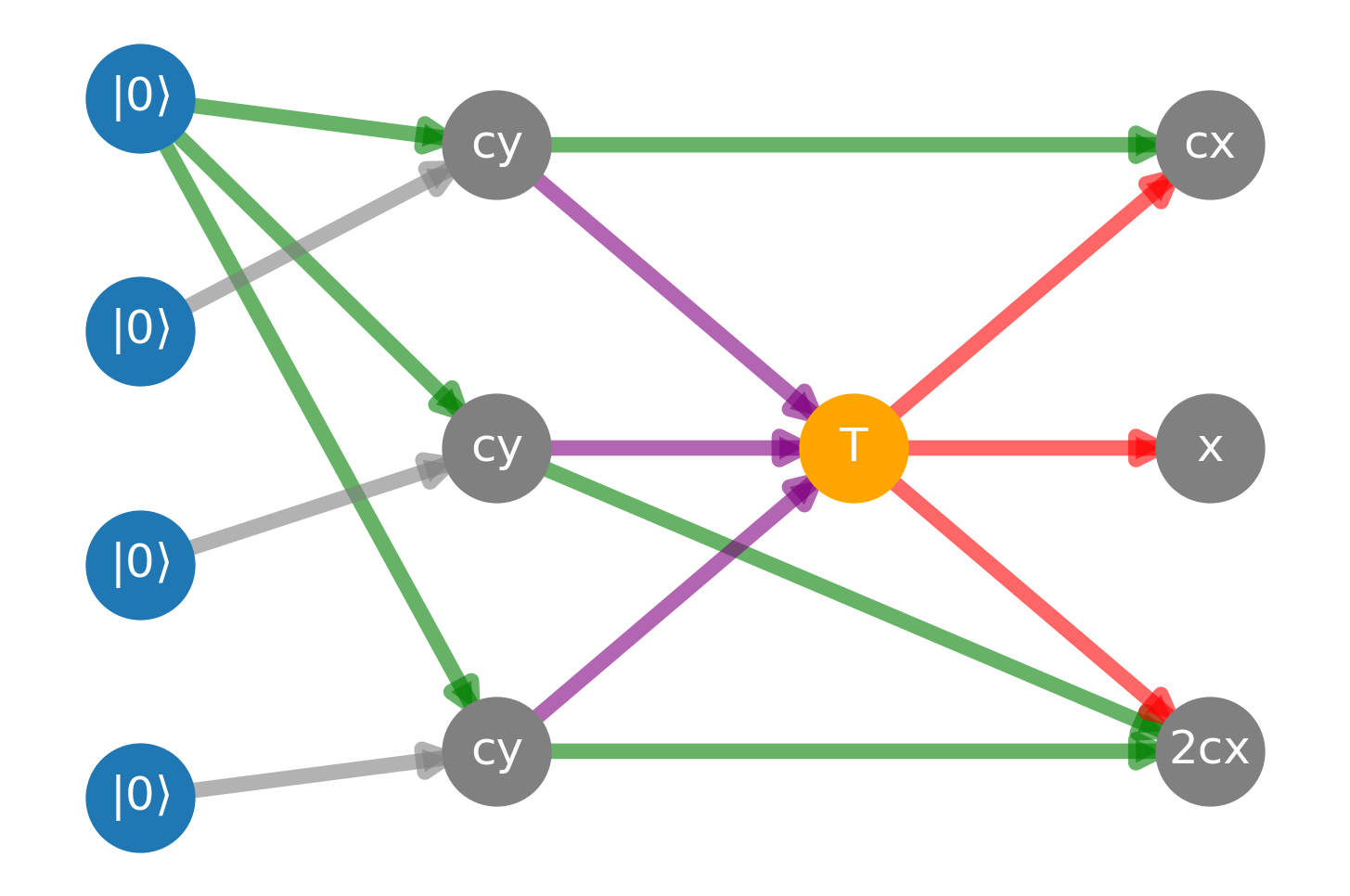}
        \caption{\label{fig:less_simple_dag}}
    \end{subfigure}
    \caption{\label{fig:dag_example}\ref{fig:simple_circuit}, \ref{fig:less_simple_circuit} Simple example circuits to demonstrate how the permeability DAG is built up. \ref{fig:simple_dag} The DAG constructed according to the rules given in Section~\ref{sec:dag}. The CY gates are Z-permeable on their control, so they form a streak. Note that these gates commute according to Theorem \ref{commutativity_theorem}, so any reordering of these gates result in a circuit with the same semantics (i.e. unitary). A reordered version of the circuit however induces the same DAG, therefore the DAG indeed represents an equivalence class of reorderings. \ref{fig:less_simple_dag} The Z-permeability streak of the CY gates is ended by applying a CX gate on $q_0$, which is X-permeable on that qubit. According to the DAG construction rules, the CX gate induces a terminator node connected with (purple) anti-depedency edges. The following streak of X-permeable operations will therefore be guaranteed to end up behind the Z-streak in a topological sort.
    For more examples please check Appendix \ref{sec:p_dag_code} for a simple snippet producing pictures like these.}
\end{figure*}

To elaborate the permeability DAG construction procedure, we start by introducing the types of nodes that can appear.
\begin{itemize}
    \item \textbf{Instruction nodes} represent quantum instructions (like a CX gate).
    \item \textbf{Allocation nodes} represent qubit resources being allocated.
    \item \textbf{Deallocation nodes} represent qubits being deallocated. The information about which qubits can be deallocated is collected during \textit{Qrisp} execution.
    \item \textbf{Terminator nodes} represent the end of a so called streak which will be elaborated shortly.
\end{itemize}
To construct a permeability DAG from a quantum circuit, we begin by adding an allocation node for each qubit in the quantum circuit. Furthemore, for every (composite) quantum gate and every deallocation we add another node to the graph. To connect the nodes we now introduce the types of edges:
\begin{itemize}
    \item \textbf{Z-edge} (green).
    \item \textbf{X-edge} (red).
    \item \textbf{Neutral edge} (grey).
    \item \textbf{Anti-dependency edge} (purple).
\end{itemize}
For each instruction node, we check the permeability status on the corresponding qubits ((de)allocation nodes are treated to be neutral) and connect the edges. Connecting the edges is done according to the following rules. Let $q$ be a qubit, that the subsequent gates $U_1,U_2$ are operating on.
\begin{enumerate}
    \item If $U_1$ and $U_2$ have differing permeability status in $q$, a Z-edge (X-edge) is created from $U_1$ to $U_2$ if $U_2$ is Z-permeable (X-permeable) in $q$. If neither is the case, a neutral edge is used instead.
    \item If $U_1, U_2$ are both Z-permeable (X-permeable), instead inserting a Z-edge (X-edge) from $U_1$ to $U_2$, we let the edge start in $U_0$, where $U_0$ is the node, which starts the edge to $U_1$ (the ``parent''). A sequence of more than one gate $U_1,..U_n$ acting on $q$ with the same permeability is therefore connected to the same node, which we call a ``streak''.
    \item If $U_{n+1}$ ends the streak $U_1,..U_n$, that is $U_{n+1}$ has a different permeability type in $q$ (or is neutral), a terminator node is inserted and every node in the streak receives an anti-dependency edge pointing towards the terminator node. Finally an edge representing the permeability of $U_{n+1}$ in $q$ is inserted pointing away from the terminator.
\end{enumerate}
To get a thorough understanding, please check the example in Fig.~\ref{fig:dag_example}.
We will now discuss the motivation behind these rules. Since the DAG will undergo a topological sorting process, rule 1 enforces that two subsequent gates acting on the same qubit will appear in the correct order upon linearization. Rule 2 is designed to reflect the commutation relations from Theorem~\ref{commutativity_theorem}. Since all gates of the streak commute, they are all connected to the same initial node $U_0$. Therefore, any valid reordering of the streak would yield the same DAG. The anti-depedency edges in rule 3 ensure that in a linearization, the final gate ending the streak will indeed always be executed last.

\subsection{Comparison to the ZX-Calculus}
We see the following similarities/differences between our construction and the ZX-Calculus\cite{Coecke_2011}:
\begin{itemize}
    \item Both representations are based around certain properties of quantum gates in the Z/X-base.
    \item The commutativity of primitive operations expressed by the permeability DAG is a special case of the Spider Fusion Rule.
    \item ZX-Calculus expressions are undirected graphs and can contain loops.
    \item The permeability DAG contains structural information in the form of (de)allocation nodes and terminator nodes that are of a fundamentally different type compared to the instruction node.
    \item The permeability DAG can abstract low level implementations and represent properties of high-level quantum functions. For instance, the commutativity in the input of subsequent adders is not immediately obvious in the ZX-calculus.
\end{itemize}

\subsection{Comparison to the Unqomp-Graph}
The following similarities/differences are found with respect to the Unqomp DAG \cite{paradis_2021}:
\begin{itemize}
    \item The Unqomp DAG only differentiates between control, target and anti-depedency edges. Conceptually, Z-edges can be identified with control edges, target edges are either X or neutral. Anti-depedency edges serve a similar purpose.
    \item With the above identifications, the permeability DAG can be used to synthesize automatic uncomputation.
    \item The permeability DAG contains the terminator node. The purpose of this node is mainly to separate two subsequent streaks from each other, which can't happen in the Unqomp DAG, since only the control edge type can have streaks.
    \item The Unqomp DAG contains allocation nodes but no deallocation nodes. For the permeability DAG, deallocation nodes play an essential role when it comes to memory management tasks.
\end{itemize}

%% file: sections/parallelization.tex
\section{Topological Sorting}
\label{sec:topo_sort}

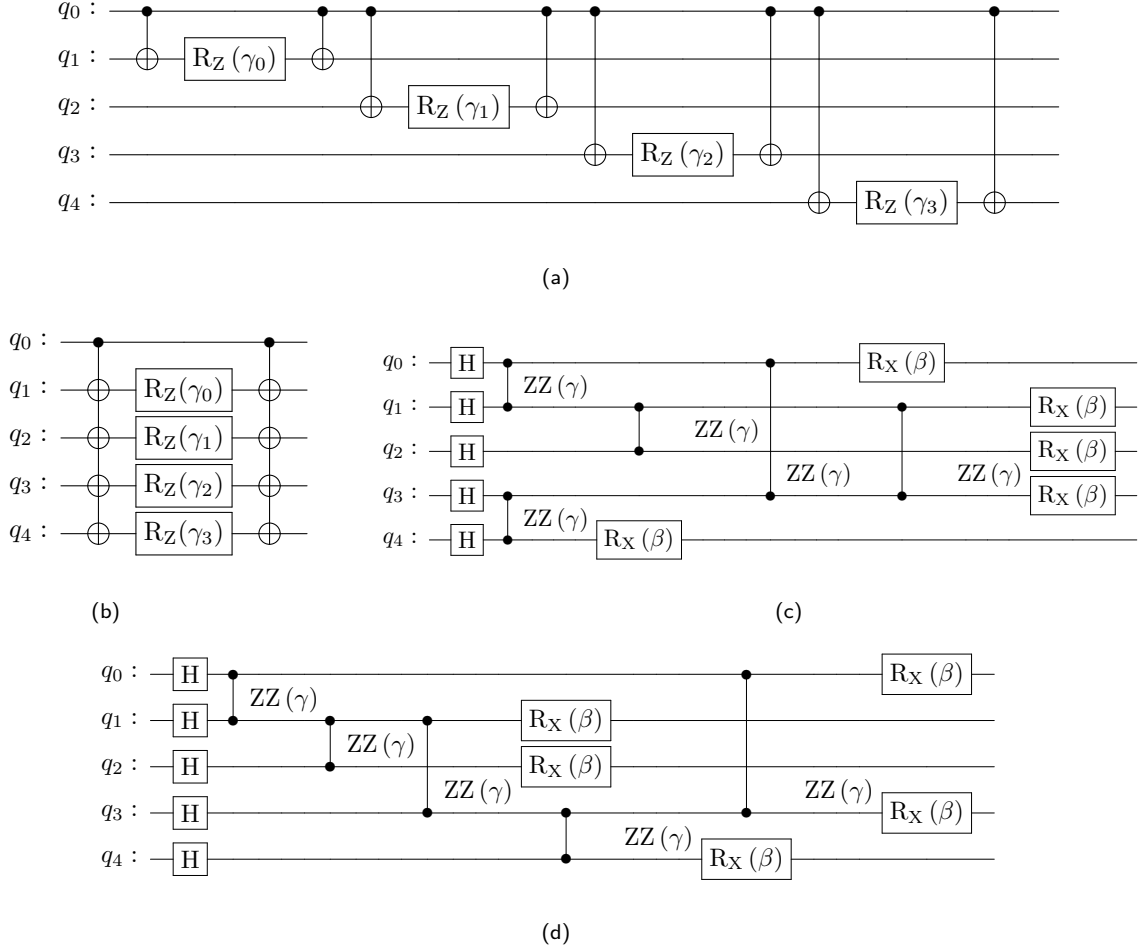
\begin{figure*}[t!]
    \begin{subfigure}[b]{1\textwidth}
        \hspace{1.5em}
        \input{graphics/rzz_chain}
        \caption{\label{fig:rzz_chain}}
    \end{subfigure}\\
    \begin{subfigure}[b]{0.3\textwidth}
        \input{graphics/rzz_chain_shrinked}
        \caption{\label{fig:rzz_chain_shrinked}}
    \end{subfigure}\hspace{1.5em}%
    \begin{subfigure}[b]{0.7\textwidth}
        \input{graphics/qaoa_circuit_par}
        \caption{\label{fig:qaoa_circuit_par}}
    \end{subfigure}\\
    \begin{subfigure}[b]{1\textwidth}
        \hspace{5em}
        \input{graphics/qaoa_circuit}
        \caption{\label{fig:qaoa_circuit}}
    \end{subfigure}
\caption{\label{fig:parallelization}\ref{fig:rzz_chain} A circuit describing a descending chain of RZZ gates, where the T-depth scales like $\mathcal{O}(n)$ in circuit size. Such circuits are a common occurence in QFT implementations or draper style adders \cite{draper_2000}. \ref{fig:rzz_chain_shrinked} The resulting circuit after applying our parallelization procedure optimizing for T-depth, leveraging the commutation relations among the CNOT gates. The T-depth is now constant regardless of the circuit scale. \ref{fig:qaoa_circuit} An example describing a $p = 1$ Max-Cut QAOA circuit. The CNOT depth is 10. \ref{fig:qaoa_circuit_par} The QAOA circuit from \ref{fig:qaoa_circuit} after applying the parallelization algorithm. The resulting CNOT depth is 6, which amounts to a 40\% improvement. Note that our implementation also works for gates with unspecified parameters, since the permeability features of the RZZ gates are independent of the particular choice of parameters. In a VQA setting, our algorithm can therefore be applied once (before the optimization loop) instead of multiple times (before each iteration).}
\end{figure*}
Now that we have the tools to build up the permeability DAG from arbitrary quantum circuits, we can start extracting equivalent reorderings of the circuit. We do this by applying a topological sorting algorithm, which is designed in such a way that certain characteristics of the circuit are improved. In this section, we describe three such algorithms for the following purposes:
\begin{enumerate}
    \item Automatic circuit parallelization in Section \ref{sec:parallelization}.
    \item Memory management in Section \ref{sec:memory_management}.
\end{enumerate}
\subsection{Parallelization}
\label{sec:parallelization}
The topological sorting algorithm for circuit parallelization is based on a adapted version of Kahn's algorithm \cite{kahn_1962}. For reader convenience, we describe the algorithm shortly:

\begin{algorithm}[h]
\caption{Kahn's Algorithm}
\KwIn{A directed acyclic graph $G = (V, E)$}
\KwResult{A list $L$ of nodes in topologically sorted order}
\SetKwFunction{Enqueue}{Enqueue}
\SetKwFunction{Dequeue}{Dequeue}

\BlankLine
\textbf{Initialize the Queue}\;
Identify all nodes with an in-degree of 0\;
Add these nodes to a queue $Q$\;

\BlankLine
\textbf{Process the Queue}\;
\While{$Q$ is not empty}{
    $node \leftarrow$ \Dequeue{$Q$}\;
    \For{each $neighbor$ connected by an outgoing edge from $node$}{
        Decrease the in-degree of $neighbor$ by 1\;
        \If{in-degree of $neighbor$ becomes 0}{
            \Enqueue{$Q$, $neighbor$}\;
        }
    }
}

\BlankLine
\textbf{Check for Cycles}\;
\If{any nodes still remain in the graph with a non-zero in-degree}{
    \textbf{Error:} The graph has a cycle, topological sort not possible\;
}

\BlankLine
\textbf{Output the Result}\;
The sequence of nodes removed from the queue represents the topological order of the graph\;

\end{algorithm}

Our hook to shaping the algorithm for our purpose lies in the dequeuing step. Since we can choose an arbitrary node from the queue, we can now deploy a heuristic for picking a node that favors circuit depth. Our strategy here is to rate all dequeuing options according to a certain cost metric $\mathcal{C}$ and choose the node with the lowest cost. We evaluated several cost metrics but focus our description on the one which produced the best results.\\
To describe the cost metric, we first need an additional concept:
\begin{tcolorbox}
\begin{definition}[Dynamic Qubit Depth]
    Given is a quantum circuit $Q$ described by a sequence of gates $(U_i)_{0\leq i < n}$ and a sequence of durations $(t_i)_{0\leq i < n}$, which indicates how long each gate would take to execute on a given physical device. The \textbf{dynamic qubit depth} of the qubit $k$ is the time $D_k(Q) \in \mathbb{R}$, which is required to execute the gate sequence until the last gate operating on qubit $k$.
\end{definition}
\end{tcolorbox}
Determining the dynamic qubit depth can be achieved with a tetris-like construction and is a straightforward task. For the purpose of parallelization, the values for each qubits can be cached and reused in the next iteration, reducing classical resource requirements significantly.\\
Building on dynamic qubit depth allows our algorithm to optimize quantum circuits not only towards a certain device class (e.g. NISQ/FT) but even towards the particular device instances by feeding the gate-timings.
\\
Using the above concepts, our cost function is
\begin{align}
\begin{aligned}
    \mathcal{C}(i, t_i, Q) = & \text{max}(D_k(Q) | k \in  \text{QB}(U_i) )\\ + &t_i/(t_{max}+1)
\end{aligned}
\end{align}
where
\begin{itemize}
    \item $Q$ is the circuit that has been produced by the previous dequeuing events.
    \item $\text{QB}(U_i)$ is the set of qubit indices that the gate $U_i$ is operating on.
    \item $t_{max} = \text{max}(t)$ is the maximum amount of time a gate can take.
\end{itemize}
We motivate the choice of these terms. Since all participating qubits of $U_i$ need to be finished with their final gate, the first term essentially determines the time when $U_i$ can be executed. In a situation where this time is equivalent for two separate gates, the second term makes sure that the faster gate is executed first.
For an example of the effect of the parallelization procedure please refer to Fig.~\ref{fig:parallelization}.
\subsubsection{Performance analysis}
The complexity of Kahn's algorithm is $\mathcal{O}(V+E)$ \cite{kahn_1962} where $V$ is the amount of vertices and $E$ is the number of edges. Applied to the permeability DAG, we infer that every qubit port for each gate can amount for at most two edges (one red/green/grey edge and one anti-dependency edge). Assuming that the amount of qubit ports per gate is bounded, we therefore deduce that the complexity for the parallelization procedure is $\mathcal{O}(N)$ where $N$ is the amount of gates, which gives our method a very favorable scaling even for problem instances with a practically relevant scale.\\
The \textit{Qrisp} implementation of the algorithm is based on the high performance computing framework \textit{Numba} \cite{numba_2015}, which ensures that even for very large circuits the parallelization step can be executed with barely any delay. Please refer to Fig.~\ref{fig:parallelization_plot} for a benchmark of the algorithm applied to QAOA circuits. Apart from the smallest instance our method improves the second best method by approximately 33\%. Furthermore, our implementation is faster by approximately one order of magnitude compared to the second fastest implementation. For further benchmarks of the technique please check the Appendix of \cite{seidel_2024_qrisp} where a Shor implementation has been optimized for T-depth.

%% file: graphics/rzz_chain.tex
\scalebox{1.0}{
\Qcircuit @C=1.0em @R=0.2em @!R { \\
	 	\nghost{{qb\_72} :  } & \lstick{{q_0} :  } & \ctrl{1} & \qw & \ctrl{1} & \ctrl{2} & \qw & \ctrl{2} & \ctrl{3} & \qw & \ctrl{3} & \ctrl{4} & \qw & \ctrl{4} & \qw & \qw\\
	 	\nghost{{qb\_73} :  } & \lstick{{q_1} :  } & \targ & \gate{\mathrm{R_Z}\,(\mathrm{\gamma_0})} & \targ & \qw & \qw & \qw & \qw & \qw & \qw & \qw & \qw & \qw & \qw & \qw\\
	 	\nghost{{qb\_74} :  } & \lstick{{q_2} :  } & \qw & \qw & \qw & \targ & \gate{\mathrm{R_Z}\,(\mathrm{\gamma_1})} & \targ & \qw & \qw & \qw & \qw & \qw & \qw & \qw & \qw\\
	 	\nghost{{qb\_75} :  } & \lstick{{q_3} :  } & \qw & \qw & \qw & \qw & \qw & \qw & \targ & \gate{\mathrm{R_Z}\,(\mathrm{\gamma_2})} & \targ & \qw & \qw & \qw & \qw & \qw\\
	 	\nghost{{qb\_76} :  } & \lstick{{q_4} :  } & \qw & \qw & \qw & \qw & \qw & \qw & \qw & \qw & \qw & \targ & \gate{\mathrm{R_Z}\,(\mathrm{\gamma_3})} & \targ & \qw & \qw\\
\\ }}

%% file: graphics/rzz_chain_shrinked.tex
\scalebox{1.0}{
\Qcircuit @C=1.0em @R=0.2em @!R { \\
	\nghost{{qb\_63} :  } & \lstick{{q_0} :  } & \ctrl{4} & \qw & \ctrl{4} & \qw\\
	\nghost{{qb\_64} :  } & \lstick{{q_1} :  } & \targ & \gate{\mathrm{R_Z}(\gamma_0)} &      \targ& \qw\\
        \nghost{{qb\_64} :  } & \lstick{{q_2} :  } & \targ & \gate{\mathrm{R_Z}(\gamma_1)} & \targ & \qw\\
        \nghost{{qb\_64} :  } & \lstick{{q_3} :  } & \targ & \gate{\mathrm{R_Z}(\gamma_2)} & \targ & \qw\\
        \nghost{{qb\_64} :  } & \lstick{{q_4} :  } & \targ & \gate{\mathrm{R_Z}(\gamma_3)} & \targ & \qw\\
\\ }}

%% file: graphics/qaoa_circuit_par.tex
\scalebox{0.9}{
\Qcircuit @C=0.9em @R=0.25em @!R { \\
	 	\nghost{{q0} :  } & \lstick{{q_0} :  } & \gate{\mathrm{H}} & \ctrl{1} & \dstick{\hspace{2.0em}\mathrm{ZZ}\,(\mathrm{\gamma})} \qw & \qw & \qw & \qw & \qw & \qw & \qw & \control \qw & \qw & \qw & \qw & \gate{\mathrm{R_X}\,(\mathrm{\beta})} & \qw & \qw & \qw & \qw & \qw \\
	 	\nghost{{q1} :  } & \lstick{{q_1} :  } & \gate{\mathrm{H}} & \control \qw & \qw & \qw & \qw & \ctrl{1} & \dstick{\hspace{2.0em}\mathrm{ZZ}\,(\mathrm{\gamma})} \qw & \qw & \qw & \qw & \qw & \qw & \qw & \ctrl{2} & \qw & \qw & \qw & \gate{\mathrm{R_X}\,(\mathrm{\beta})} & \qw \\
	 	\nghost{{q2} :  } & \lstick{{q_2} :  } & \gate{\mathrm{H}} & \qw & \qw & \qw & \qw & \control \qw & \qw & \qw & \qw & \qw & \dstick{\hspace{2.0em}\mathrm{ZZ}\,(\mathrm{\gamma})} \qw & \qw & \qw & \qw & \dstick{\hspace{2.0em}\mathrm{ZZ}\,(\mathrm{\gamma})} \qw & \qw & \qw & \gate{\mathrm{R_X}\,(\mathrm{\beta})} & \qw \\
	 	\nghost{{q3} :  } & \lstick{{q_3} :  } & \gate{\mathrm{H}} & \ctrl{1} & \dstick{\hspace{2.0em}\mathrm{ZZ}\,(\mathrm{\gamma})} \qw & \qw & \qw & \qw & \qw & \qw & \qw & \ctrl{-3} & \qw & \qw & \qw & \control \qw & \qw & \qw & \qw & \gate{\mathrm{R_X}\,(\mathrm{\beta})} & \qw \\
	 	\nghost{{q4} :  } & \lstick{{q_4} :  } & \gate{\mathrm{H}} & \control \qw & \qw & \qw & \qw & \gate{\mathrm{R_X}\,(\mathrm{\beta})} & \qw & \qw & \qw & \qw & \qw & \qw & \qw & \qw & \qw & \qw & \qw & \qw & \qw \\
\\ }}

%% file: graphics/qaoa_circuit.tex
\scalebox{0.95}{
\Qcircuit @C=0.9em @R=0.23em @!R { \\
	 	\nghost{{q0} :  } & \lstick{{q_0} :  } & \gate{\mathrm{H}} & \ctrl{1} & \dstick{\hspace{2.0em}\mathrm{ZZ}\,(\mathrm{\gamma})} \qw & \qw & \qw & \qw & \qw & \qw & \qw & \qw & \qw & \qw & \qw & \qw & \qw & \qw & \qw & \control \qw & \qw & \qw & \qw & \gate{\mathrm{R_X}\,(\mathrm{\beta})} & \qw \\
	 	\nghost{{q1} :  } & \lstick{{q_1} :  } & \gate{\mathrm{H}} & \control \qw & \qw & \qw & \qw & \ctrl{1} & \dstick{\hspace{2.0em}\mathrm{ZZ}\,(\mathrm{\gamma})} \qw & \qw & \qw & \ctrl{2} & \qw & \qw & \qw & \gate{\mathrm{R_X}\,(\mathrm{\beta})} & \qw & \qw & \qw & \qw & \qw & \qw & \qw & \qw & \qw \\
	 	\nghost{{q2} :  } & \lstick{{q_2} :  } & \gate{\mathrm{H}} & \qw & \qw & \qw & \qw & \control \qw & \qw & \qw & \qw & \qw & \dstick{\hspace{2.0em}\mathrm{ZZ}\,(\mathrm{\gamma})} \qw & \qw & \qw & \gate{\mathrm{R_X}\,(\mathrm{\beta})} & \qw & \qw & \qw & \qw & \dstick{\hspace{2.0em}\mathrm{ZZ}\,(\mathrm{\gamma})} \qw & \qw & \qw & \qw & \qw \\
	 	\nghost{{q3} :  } & \lstick{{q_3} :  } & \gate{\mathrm{H}} & \qw & \qw & \qw & \qw & \qw & \qw & \qw & \qw & \control \qw & \qw & \qw & \qw & \ctrl{1} & \dstick{\hspace{2.0em}\mathrm{ZZ}\,(\mathrm{\gamma})} \qw & \qw & \qw & \ctrl{-3} & \qw & \qw & \qw & \gate{\mathrm{R_X}\,(\mathrm{\beta})} & \qw \\
	 	\nghost{{q4} :  } & \lstick{{q_4} :  } & \gate{\mathrm{H}} & \qw & \qw & \qw & \qw & \qw & \qw & \qw & \qw & \qw & \qw & \qw & \qw & \control \qw & \qw & \qw & \qw & \gate{\mathrm{R_X}\,(\mathrm{\beta})} & \qw & \qw & \qw & \qw & \qw \\
\\ }}

%% file: sections/memory_management.tex
\subsection{Memory Management}
\label{sec:memory_management}
This section treats the problem of reordering the operation sequence such that the amount of required Qubits is optimized. To present this feature, the section is divided into three subsections. 
\begin{itemize}
    \item In subsection~\ref{sec:why_order}, we elaborate how and why reordering a circuit can have influence on the amount of required qubits. 
    \item Similar to the section on parallelization, in subsection~\ref{sec:flex_sort}, we construct a specialized algorithm for topological sorting called Flex-Sort, which caters to our requirements.
    \item In subsection~\ref{sec:dealloc_order}, we demonstrate how Flex-Sort is applied to the permeability DAG.
\end{itemize}

\subsubsection{Why order matters}
\label{sec:why_order}
Within \textit{Qrisp} the allocation events are triggered by calling the \texttt{QuantumVariable} constructor and deallocations can be achieved by calling either the \texttt{.delete} or \texttt{.uncompute} method. On the level of the intermediate representation, (de)allocations are treated as a particular kind of operation among the regular quantum gates. To elaborate how topological sorting can help to acquire a good allocation strategy, we introduce the following concepts:
\begin{itemize}
    \item A \textbf{algorithmic qubit} is a qubit which takes a particular role in an algorithm. Algorithmic qubits are (de)allocated once.
    \item An \textbf{execution qubit} is a qubit in the optimized quantum circuit and can host multiple algorithmic Qubits if their lifetime falls in distinct steps of the algorithm.
\end{itemize}
To map a sequence of gates $S$ operating on algorithmic qubits to a quantum circuit with execution qubits, we apply the following procedure:
\begin{enumerate}
    \item Analyze $S$ regarding how many algorithmic qubits are allocated during peak load.
    \item Create a new quantum circuit with that many execution qubits and successively insert the corresponding quantum gates. During this procedure the pool of available execution qubits is managed in the following way\footnote{By creating a circuit with even more than the required execution qubits, it is possible to give the compiler the opportunity to choose the allocated qubit from a bigger pool of choices. This can significantly reduce circuit depth as the load can be distributed. This feature is described as ``workspace'' within \cite{seidel_2024_qrisp}.}
    \begin{enumerate}
        \item The \textbf{allocation} of an algorithmic qubit $q_{alg}$ chooses an execution qubit $q_{ex}$ from the pool according to a certain heuristic\footnote{To decide about the most suitable execution qubit, we leverage the concept of dynamic qubit depth (defined earlier) to determine which of the available execution qubits would be free the earliest in an actual execution of the quantum circuit that has been compiled so far.}. All subsequent operations on $q_{alg}$ are executed on $q_{ex}$ until deallocation.
        \item The \textbf{deallocation} of an algorithmic qubit $q_{alg}$ returns the corresponding $q_{ex}$ to the pool.
    \end{enumerate}
\end{enumerate}
From this construction we see that the amount of required qubits can be influenced by changing the order of allocations. In particular, if a deallocation event can be ``pulled in front'' of an allocation event, which would trigger peak load, the amount of required qubits decreases by one. \\
\subsubsection{Flex-Sort}
\label{sec:flex_sort}
Our goal is therefore to find a reordering of the gate sequence such that deallocations are executed as early as possible and allocations are executed as late as possible. We achieve this reformulating another strategy for DAG linearization: Depth-first traversal \cite{intro_to_algos_2009}.
\begin{algorithm}[t]
\caption{Depth First Topological Sort}
\SetKwFunction{Fvisit}{visit}
\SetKwProg{Fn}{Function}{:}{}
\KwData{A directed acyclic graph G = (E,V)}
\KwResult{A list $L$ of nodes in topologically sorted order}

$L \leftarrow$ empty list to store the sorted nodes\;
\While{nodes without a permanent mark exist}{
    select an arbitrary unmarked node $n$\;
    \Fvisit{$n$}\
}
\Return{L}

\Fn{\Fvisit{node $n$}}{
    \If{$n$ has a permanent mark}{
        \KwRet
    }
    \If{$n$ has a temporary mark}{
        \textbf{stop}
        {graph contains a cycle\;}
    }

    mark $n$ with a temporary mark\;
    
    \For{each $m \in V, (n,m) \in E$}{
        \Fvisit{$m$}\;
    }
    mark $n$ with a permanent mark\;
    add $n$ to head of $L$\;
}
\end{algorithm}

To make this algorithm useful for our purposes, we note that the \texttt{visit} function in its essence creates a linearization of the subgraph $\text{\texttt{desc}}(G, n)$\footnote{$\texttt{desc}(G,n)$ stands for the descendants subgraph of $G$ in $n$ and comprises all nodes that are reachable starting in $n$} and inserts it at the head of the result list L. Instead of using multiple recursions of the \texttt{visit} function, we now generalize the algorithm such that an arbitrary blackbox topological sorting (TS) algorithm can be used as a ``backend''.

\begin{algorithm}[t]
\caption{Flex-Sort}
\SetKwFunction{Fvisit}{visit}
\SetKwProg{Fn}{Function}{:}{}
\KwData{A directed acyclic graph G = (E,V)\\
An arbitrary topological sort algorithm TS}
\KwResult{A list $L$ of nodes in topologically sorted order}

$L \leftarrow$ empty list to store the sorted nodes\;
\While{G contains at least one node}{
    select an arbitrary node $n$\;
    $K \leftarrow$ TS(\texttt{anct}(G, n))\;
    extend L by K\;
    remove all K from G\;
}
\Return{L}
\end{algorithm}
Note that instead of using the descendants subgraph, we used the ancestors, which is essentially the descendants of the transposed graph\footnote{The transpose of a directed graph $G$ is the graph that is acquired when flipping all the edges in the opposite direction.}. To make up for this deviation we \textbf{extend} the resulting list L instead of inserting the entries at the head of L. We denote this algorithm ``Flex-Sort'' because it gives us the following flexibility during execution:
\begin{itemize}
    \item We can choose the backend TS.
    \item For each iteration, we can choose a suitable initial node $n$.
\end{itemize}
An obvious choice for TS is the parallelization algorithm\footnote{In our implementation we avoid callig the parallelization algorithm $|D|$ times. Instead we call the parallelization algorithm before memory management once and assign each node an integer label corresponding to the nodes position in the linearization. During memory management, this label is then used as a sorting key for a regular sorting algorithm to emulate a topological sort.}, however any other topological sorting algorithm is possible too.
As elaborated above, it is our goal to execute deallocation nodes as early as possible, which is why we select this node type as the prioritized starting nodes.

\begin{figure*}[t!]
    \begin{subfigure}[b]{0.5\textwidth}
        \centering
        \input{graphics/mm_circuit}
        \caption{\label{fig:mm_circuit}}
    \end{subfigure}%
    \begin{subfigure}[b]{0.5\textwidth}
        \centering
        \includegraphics[width = 0.7 \textwidth]{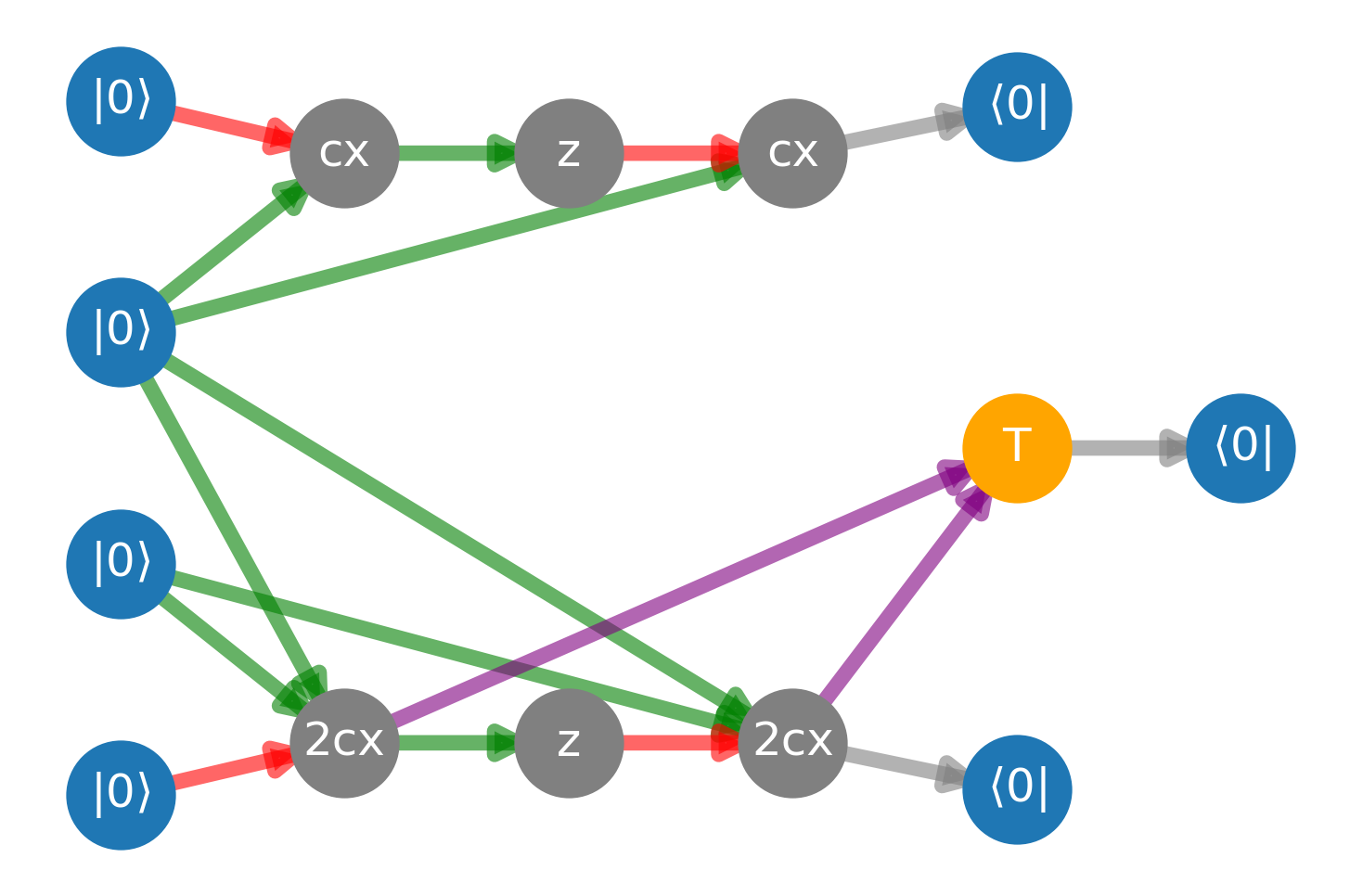}
        \caption{\label{fig:mm_dag}}
    \end{subfigure}
    \\
    \\
    \centering
    \begin{subfigure}[c]{\textwidth}
        \hspace{7em}
        \input{graphics/mm_compressed_circuit}
        \caption{\label{fig:mm_circuit_compressed}}
    \end{subfigure}
    \caption{\label{fig:mm_example}\ref{fig:mm_circuit} A simple circuit containing some deallocations (marked by the $\bra{0}$ gates). \ref{fig:mm_dag} The permeability DAG of \ref{fig:mm_circuit}. Note that the ancestors of the highest deallocation node contain only two allocation nodes. The ancestors of the lower two contain three. Based on this information, our compiler chooses to perform the necessary quantum gates to perform the upper allocation as early as possible, allocating only the minimal amount of qubits. \ref{fig:mm_circuit_compressed} The resulting circuit after applying the memory management topological sorting algorithm to the permeability DAG in \ref{fig:mm_dag}.}
\end{figure*}

\subsubsection{Deallocation order}
\label{sec:dealloc_order}
But how to choose the order of deallocation nodes that suits our efforts the most? For that we investigate the DAG a bit more (prior to sorting). In particular, we generate the \textit{ancestors} subgraph for each \textit{deallocation} node and count how many \textit{allocations} are required to execute that particular deallocation. The deallocation nodes are then ranked according to the amount of allocations that are required to execute them. 

We summarize our procedure: Given is a sequence of operations $S$, a set of allocation events $A \subset S$, a set of deallocation events $D \subset S$.
The following steps are taken to generate an equivalent reordering, which executes the deallocations as early as possible and the allocations as late as possible.
\begin{enumerate}
    \item Generate the permeability DAG $G$ of $S$
    \item For each deallocation node $d$ determine the ancestors subgraph \texttt{anct}$(G, d)$ and count how many allocations are contained in that subgraph. Denote this amount $|A(d)|$.
    \item Create a list of deallocation nodes and sort it according to the sorting key $|A(d)|$.
    \item Execute the Flex-Sort algorithm by picking the nodes of the list from the previous step as initial nodes.
    \item Iterate through the sequence generated in the previous step and dynamically (de)allocate execution qubits as described in Section \ref{sec:why_order}.
\end{enumerate}
For an example of this procedure please consider Fig.~\ref{fig:mm_example}. Benchmarking this algorithm is difficult since benchmark sets like \cite{quetschlich_2023} contain no deallocation information. In practice we could verify for a variety of examples that our algorithm finds an allocation strategy requiring the optimal amount of qubits. This is however not always the case.
\subsubsection{Performance analysis}
The above described procedure for establishing a memory management strategy is more demanding than the parallelization technique. It is however not the topological sorting itself, which is more expensive\footnote{This is because the overall amount of edges can only be lower if the graph is cut into several ancestor subgraphs.} but instead determining the order of deallocation nodes, because this requires computing the ancestors graph for every deallocation node. In the worst case the complexity is therefore $\mathcal{O}(|S|\cdot |D|)$. However this task can be parallelized (one thread for each deallocation node). Within 
\textit{Qrisp} this bottleneck is remedied by a parallel Numba \cite{numba_2015} kernel, which performs our technique with barely any delay for most mid-term quantum algorithms. Additionally, the ancestors function is available in the cuGraph framework \cite{kang_2022}, which allows to outsource this task to clusters of GPU hardware, implying even circuits with billions of gates might be treatable.

%% file: graphics/mm_circuit.tex
\scalebox{1.0}{
\Qcircuit @C=1.0em @R=0.2em @!R { \\
	 	\nghost{{qb\_287} :  } & \lstick{{q_0} :  } & \targ & \gate{\mathrm{Z}} & \qw & \qw & \targ & \gate{\bra{0}} & \qw & \qw\\
	 	\nghost{{qb\_288} :  } & \lstick{{q_1} :  } & \ctrl{-1} & \ctrl{1} & \qw & \ctrl{1} & \ctrl{-1} & \qw & \qw & \qw\\
	 	\nghost{{qb\_289} :  } & \lstick{{q_2} :  } & \qw & \ctrl{1} & \qw & \ctrl{1} & \gate{\bra{0}} & \qw & \qw & \qw\\
	 	\nghost{{qb\_290} :  } & \lstick{{q_3} :  } & \qw & \targ & \gate{\mathrm{Z}} & \targ & \gate{\bra{0}} & \qw & \qw & \qw\\
\\ }}

%% file: graphics/mm_compressed_circuit.tex
\scalebox{1.0}{
\Qcircuit @C=1.0em @R=0.2em @!R { \\
	 	\nghost{{qbl.0} :  } & \lstick{q_0 :  } & \ctrl{1} & \qw & \ctrl{1} & \qw & \ctrl{1} & \qw & \ctrl{1} & \qw & \qw & \qw\\
	 	\nghost{{workspace\_0} :  } & \lstick{{q_1} :  } & \targ & \gate{\mathrm{Z}} & \targ & \gate{\bra{0}} & \targ & \gate{\mathrm{Z}} & \targ & \gate{\bra{0}} & \qw & \qw\\
	 	\nghost{{workspace\_1} :  } & \lstick{q_2 :  } & \qw & \qw & \qw & \qw & \ctrl{-1} & \qw & \ctrl{-1} & \gate{\bra{0}} & \qw & \qw\\
\\ }}

%% file: sections/outlook.tex
\section{Outlook}
Even though the present article describes the permeability DAG in terms of parallelization and memory management, more use-cases are possible:
\begin{enumerate}
    \item \textbf{Light-cone reduction}. Viewing the permeability DAG as a representation of causally related events, it is possible to filter out gates which do not influence measurement results. In a different words: It doesn't matter whether a gate outside of the light-cone of a measurement is executed. We can use this fact to optimize quantum circuits in that we remove every node from the permeability DAG that is outside of the light-cone of the intended measurements.
    \item \textbf{Peephole optimizations}. The permeability DAG can help with identifying peephole optimizations that require adjacency. One example of this is the simple gate sequence $(Z, S, Z)$. The Z gates cancel out but only because the sequence can be arbitrarily reordered due to permeability features. Less trivial examples could include cancellation/merging of CX/CP gates or even higher level semantics such as fusing two quantum adders.
    \item \textbf{Dynamic parallelization}. For circuits that contain very long Z-permeable streaks\footnote{An example for this phenomenon is the controlled call of a composite quantum function. Every subroutine of this function is individually controlled such that control qubit access is a bottleneck.} it is possible to ``quasi-copy'' the value of the streak qubit and therefore enable the access to the streak qubit value via multiple qubits, facilitating an improved circuit depth. Quasi-copy means duplicating a computational basis state by applying a CX gate into a freshly allocated qubit.
\end{enumerate}
Apart from the mentioned potential use-cases, further questions regarding the permeability DAG remain open:\\
Within Section~6 of \cite{seidel_2024_qrisp}, we describe our efforts to restructure \textit{Qrisp} to remedy compilation speed bottlenecks but also to enable seamless hybrid algorithm compilation. For this we leverage the Jax \cite{jax} framework, which exposes an interface for an extensible dynamic intermediate representation called \textit{Jaxpr}. Quantum algorithms expressed through Jaxprs can not only contain classical control flow (such as loops or conditionals) but also purely classical functions. As such we are faced with the question how these Jaxpr features can be embedded into the permeability DAG framework to facilitate the described compilation algorithms also for hybrid settings.

%% file: sections/summary.tex
\section{Summary}
Within the present work we gave a rigorous definition of the concepts of Z/X-permeability, including a proof that permeable gates commute if they only intersect on permeable inputs. To make use of the induced commutation relations, we constructed the permeability DAG, which resembles the DAG defined in the Unqomp algorithm \cite{paradis_2021}. Next to the synthesis of uncomputation, the permeability DAG can be used to find equivalent reorderings of the circuit by applying a topological sorting algorithm. We describe two such sorting algorithms, which are constructed to optimize certain features of the resulting quantum circuit. The first algorithm parallelizes quantum function calls to reduce the depth of the overall circuit, whereas the second method reorders the program instructions to reduce peak quantum memory consumption. For both of these algorithms we gave a rough complexity analysis. Our implementation in the \textit{Qrisp} framework shows that for near/mid-term quantum algorithms, the introduced delay of the procedures is negligible. Finally, we elaborated on further applications of the permeability DAG and discussed its implementation based on a more dynamic IR.

%% file: sections/appendix.tex
\section*{Appendix}
\section{Proof of Theorem \ref{commutativity_theorem}}
\label{sec:x_com_proof}
This appendix contains the proofs for Theorem~\ref{commutativity_theorem} from Section~\ref{sec:permeability}. In order to prove Theorem~\ref{commutativity_theorem}, we first need the following theorem.

\begin{tcolorbox}
\begin{theorem}
    \label{expansion_theorem}
    Let $U \in U(2^n)$ be an $n$-qubit operator. If $U$ is Z-permeable on the first $p$ qubits, there are operators $\tilde{U}_0, \tilde{U}_1,\dotsc, \tilde{U}_{2^p -1}$ such that:
    \begin{align}
        \label{expansion}
        U = \sum_{i = 0}^{2^p-1} \ket{i}\bra{i} \otimes \tilde{U}_i.
    \end{align}
\end{theorem}
\end{tcolorbox}

\begin{proof}
    We will treat the case $p = 1$ first and generalize via induction afterwards.\\
    We start by inserting identity operators $\mathbbm{1} = \sum_{i = 0} \ket{i}\bra{i}$:
    \begin{align}
        U = &\mathbbm{1} U \mathbbm{1} \\
        =& \sum_{i,j = 0}^1 \ket{i}\bra{i} U \ket{j}\bra{j}\\
        =& \sum_{i,j = 0}^1 \ket{i} \bra{j} \otimes \hat{U}_{ij}.
    \end{align}
    where $\hat{U}_{ij} = \bra{i} U \ket{j}$.
    Due to the Z-permeability condition, we have
    \begin{align}
    \begin{aligned}
        0 =& Z_0 U - U Z_0\\
        = & \left( \sum_{k =0}^1 (-1)^k \ket{k}\bra{k} \otimes \mathbbm{1}^{\otimes{n-1}}\right) \left(\sum_{i,j = 0}^1 \ket{i} \bra{j} \otimes \hat{U}_{ij}\right)\\
        -& \left( \sum_{i,j = 0}^1 \ket{i} \bra{j} \otimes \hat{U}_{ij}\right) \left(\sum_{k = 0}^1 (-1)^k \ket{k}\bra{k} \otimes \mathbbm{1}^{\otimes{n-1}}\right)\\
        = & \sum_{i,j,k = 0}^1 (-1)^k \left(\ket{k} \langle k | i \rangle \bra{j} \otimes \hat{U}_{ij} - \ket{i} \langle j | k \rangle \bra{k} \otimes \hat{U}_{ij}\right)\\
        = & \sum_{i,j,k = 0}^1 (-1)^k \left(\ket{k} \langle k | i \rangle \bra{j} - \ket{i} \langle j | k \rangle \bra{k}\right) \otimes \hat{U}_{ij}\\
        = & \sum_{i,j = 0}^1 \left((-1)^i \ket{i} \bra{j} - (-1)^j \ket{i} \bra{j}\right) \otimes \hat{U}_{ij}.
    \end{aligned}
    \end{align}
    From this form, we see that the index constellations, where $i = j$ cancel out. We end up with
    \begin{align}
        0 = 2 (\ket{0}\bra{1} \otimes \hat{U}_{01} - \ket{1}\bra{0} \otimes \hat{U}_{10}).
    \end{align}
    Since both summands act on disjoint subspaces, we conclude
    \begin{align}
        \hat{U}_{01} = 0 = \hat{U}_{10}.
    \end{align}
    Finally, we set
    \begin{align}
    \begin{aligned}
        \tilde{U}_0 = \hat{U}_{00}\\
        \tilde{U}_1 = \hat{U}_{11}
    \end{aligned}
    \end{align}
    yielding the claim for p = 1.
    To complete the proof we give the induction step, that is, we prove the claim for $p = p_0 + 1$ under the assumption that it is true for $p = p_0$:
    Since $U$ is permeable on qubit $p_0 + 1$, we have
    \begin{align}
        0 =& Z_{p_0 + 1} U - U Z_{p_0 + 1}
    \end{align}
    As the claim is true for $p = p_0$, we insert
    \begin{align}
        \label{small_expansion}
        U = \sum_{i = 0}^{2^{p_0} - 1} \ket{i}\bra{i} \otimes \tilde{U}_i
    \end{align}
    yielding
    \begin{align}
        0 = \sum_{i = 0}^{2^{p_0} - 1} \ket{i}\bra{i} \otimes (Z_{p_0 +1} \tilde{U}_i - Z_{p_0 +1} \tilde{U}_i )
    \end{align}
    Since each of the summand operators acts on disjoint subspaces, we conclude
    \begin{align}
        0 =  Z_{p_0 +1} \tilde{U}_i - Z_{p_0 +1} \tilde{U}_i
    \end{align}
    This, as shown above, implies
    \begin{align}
        \tilde{U}_i = \sum_{j = 0}^{1} \ket{j}\bra{j} \otimes (\tilde{U}_i)_j.
    \end{align}
    Finally, we insert this form into eq.~\ref{small_expansion} to retrieve the claim for $p = p_0 + 1$:
    \begin{align}
        U = \sum_{i = 0}^{2^{p_0+1} - 1} \ket{i}\bra{i} \otimes \tilde{U}_i
    \end{align}
\end{proof}

Having proved the above theorem, the next step is to employ it in the proof
of Theorem~\ref{commutativity_theorem} for Z-permeability.

\begin{proof}
According to Theorem \ref{expansion_theorem} we can write
\begin{align}
    (U \otimes \mathbbm{1}^{\otimes m-p}) = \sum_{i = 0}^{2^p-1} \tilde{U}_i \otimes \ket{i}\bra{i} \otimes  \mathbbm{1}^{\otimes m-p}\\
    (\mathbbm{1}^{\otimes n-p} \otimes V) = \sum_{j = 0}^{2^p-1} \mathbbm{1}^{\otimes n-p} \otimes \ket{j}\bra{j} \otimes \tilde{V}_j
\end{align}
Multiplying these operators gives
\begin{align}
\begin{aligned}
    &(U \otimes \mathbbm{1}^{\otimes m-p}) (\mathbbm{1}^{\otimes n-p} \otimes V)\\
    = &\left( \sum_{i = 0}^{2^p-1} \tilde{U}_i \otimes \ket{i}\bra{i} \otimes \mathbbm{1}^{\otimes m-p} \right)\\&\left(\sum_{j = 0}^{2^p-1} \mathbbm{1}^{\otimes n-p} \otimes \ket{j}\bra{j} \otimes \tilde{V}_j\right)\\
    = &\sum_{i,j = 0}^{2^p-1}  \tilde{U}_i \otimes \ket{i} \langle i|j \rangle  \bra{j} \otimes \tilde{V}_j\\
    = &\sum_{i = 0}^{2^p-1}  \tilde{U}_i \otimes \ket{i} \bra{i} \otimes \tilde{V}_i
\end{aligned}
\end{align}
Multiplication in reverse order yields the same result:
\begin{align}
\begin{aligned}
    &(\mathbbm{1}^{\otimes n-p} \otimes V) (U \otimes \mathbbm{1}^{\otimes m-p})\\
    = & \left(\sum_{j = 0}^{2^p-1} \mathbbm{1}^{\otimes n-p} \otimes  \ket{j}\bra{j} \otimes 
 \tilde{V}_j\right)\\ & \left(\sum_{i = 0}^{2^p-1} \tilde{U}_i \otimes  \ket{i}\bra{i} \otimes 
 \mathbbm{1}^{\otimes m-p}\right) \\
    = &\sum_{i,j = 0}^{2^p-1} \tilde{U}_i \otimes  \ket{j} \langle j|i \rangle  \bra{i} \otimes \tilde{V}_j\\
    = &\sum_{i = 0}^{2^p-1}  \tilde{U}_i \otimes  \ket{i} \bra{i} \otimes \tilde{V}_i
\end{aligned}
\end{align}
From this we conclude the claim.
\end{proof}

Finally, we generalize the proof from just Z-permeability to X-permeability.
\begin{proof}
    Let $U, V$ satisfy the given conditions for X-permeability. We define two auxiliary operators
    \begin{align}
        \tilde{U} &= H_{\geq n-p}  U H_{\geq n-p}\\
        \tilde{V} &= H_{<p} V H_{<p}
    \end{align}
    The notation $H_{<p}$ here implies that Hadamard gates are applied to all qubits with an index $<p$.
    We observe that $\tilde{U}, \tilde{V}$ are both Z-permeable on the relevant qubits. To see this, let $0\leq k< p$:
    \begin{align}
    \begin{aligned}
        Z_k \tilde{V} &= H_k X_k H_k \tilde{V}\\
        &=  H_k X_k H_{<p, \neq k} V H_{<p}\\
        &= H_k H_{<p, \neq k} X_k V H_{<p}\\
        &= H_{<p} V X_k H_k H_{<p, \neq k}\\
        &= H_{<p} V H_{<p} H_k X_k H_k\\
        &= H_{<p} V H_{<p} Z_k\\
        &= \tilde{V} Z_k
    \end{aligned}
    \end{align}
    Obviously, a similar reasoning holds vor $\tilde{U}$.
    Using the Z-permeability commutativity theorem, we deduct that $\tilde{U}, \tilde{V}$ commute if they only intersect on the relevant qubits.
    \begin{align}
    \begin{aligned}
        (\tilde{U} \otimes \mathbbm{1}^{\otimes m-p})& (\mathbbm{1}^{\otimes n-p} \otimes \tilde{V})\\
        &=\\
        (\mathbbm{1}^{\otimes n-p} \otimes \tilde{V})& (\tilde{U} \otimes \mathbbm{1}^{\otimes m-p})
    \end{aligned}
    \end{align}
    By wrapping both sides of the equation into Hadamard gates, we obtain the statement:
        \begin{align}
        \begin{aligned}
        H_{\geq n-p, <n}(\tilde{U} \otimes \mathbbm{1}^{\otimes m-p})& (\mathbbm{1}^{\otimes n-p} \otimes \tilde{V})H_{\geq n-p, <n}\\
        &=\\
        H_{\geq n-p, <n}(\mathbbm{1}^{\otimes n-p} \otimes \tilde{V})& (\tilde{U} \otimes \mathbbm{1}^{\otimes m-p})H_{\geq n-p, <n}
        \\ &\Leftrightarrow\\
        (U \otimes \mathbbm{1}^{\otimes m-p})& (\mathbbm{1}^{\otimes n-p} \otimes V)\\
        &=\\
        (\mathbbm{1}^{\otimes n-p} \otimes V)& (U \otimes \mathbbm{1}^{\otimes m-p})
        \end{aligned}
    \end{align}

\end{proof}

\section{Permeability Graph Code}
\label{sec:p_dag_code}
In the following, we provide some \textit{Qrisp} code to reproduce or alter the plots in Fig.~\ref{fig:dag_example}.

\begin{minted}{python}
from qrisp import *
qc = QuantumCircuit(4)

# Z-permeable streak
qc.cy(0,1)
qc.cy(0,2)
qc.cy(0,3)

# X-permeable streak
qc.cx(1,0)
qc.x(0)
qc.mcx([2,3],0)

dag = PermeabilityGraph(qc)
dag.draw()
\end{minted}
As a subclass of \texttt{networkx.DiGraph}, the permeability graph object can be processed by a variety of Networkx algorithms.